\documentclass[letterpaper, 11pt]{article}
\usepackage[margin=1in]{geometry}
\usepackage{fancyhdr}

\usepackage[authoryear, round]{natbib}
\usepackage{times}
\usepackage{multirow}
\usepackage{booktabs}
\usepackage{mathtools}
\usepackage{thmtools}
\usepackage{bm}
\usepackage{microtype}
\usepackage[colorlinks=true,linkcolor=blue!50!black,citecolor=blue!50!black,
            urlcolor=blue!50!black]{hyperref}
\usepackage{cleveref}
\usepackage{amsmath,amssymb, amsthm}
\usepackage{tikz}
\usepackage{authblk}
\usetikzlibrary{arrows.meta,positioning}

\theoremstyle{plain}
\newtheorem{theorem}{Theorem}[section]

\newtheorem{lemma}[theorem]{Lemma}

\theoremstyle{definition}
\newtheorem{definition}[theorem]{Definition}

\theoremstyle{remark}

\newcommand{\FF}{\mathbb{F}}
\newcommand{\R}{\mathbb{R}}

\newcommand{\EE}{\mathbb{E}}
\newcommand{\E}{\mathbb{E}}
\newcommand{\PR}{\Pr}
\newcommand{\Unif}{\mathrm{Unif}}
\newcommand{\Topk}{\mathrm{Top}_k}

\newcommand{\STAT}{\mathrm{STAT}}
\newcommand{\Var}{\mathrm{Var}}

\newcommand{\rec}{\mathrm{recall}}

\newenvironment{CompactItemize}{
\begin{list}{\tiny$\bullet$}{%
\setlength{\leftmargin}{10pt}
\setlength{\itemindent}{0pt}
\setlength{\topsep}{-1pt}
\setlength{\itemsep}{0pt}
}}
{\end{list}}

\newcommand{\xorpanel}[9]{%
  \begin{tikzpicture}[
    every node/.style={font=\small},
    kw/.style={draw, rounded corners=1pt, minimum width=20mm,
               minimum height=6mm, align=center, font=\small},
    gate/.style={draw, circle, inner sep=0.6pt, font=\small},
    bit/.style={draw, minimum size=5.5mm, font=\small},
    hi/.style={draw, line width=1.1pt, minimum size=5.5mm, font=\small},
    doc/.style={draw, line width=1.1pt, minimum width=13mm,
                minimum height=7.5mm, font=\small},
    arr/.style={-{Latex[length=1.8mm]}, gray, semithick},
    emit/.style={-{Latex[length=2.2mm]}, semithick},
  ]
    \node[font=\bfseries] at (2.05,2.95) {#1};

    \node[kw, fill=#2] (w) at (0,2)  {waterproof};
    \node[kw, fill=#3] (i) at (0,1)  {insulated};
    \node[kw, fill=#4] (l) at (0,0)  {lightweight};

    \node[gate] (g1) at (1.95,2) {$\oplus$};
    \node[gate] (g2) at (1.95,1) {$\oplus$};
    \node[gate] (g3) at (1.95,0) {$\oplus$};

    \draw[arr] (w.east) -- (g1.180);
    \draw[arr] (i.east) -- (g1.180);
    \draw[arr] (w.east) -- (g2.180);
    \draw[arr] (l.east) -- (g2.180);
    \draw[arr] (i.east) -- (g3.180);
    \draw[arr] (l.east) -- (g3.180);

    \node[bit]  (b1) at (3.05,2) {#5};
    \node[bit] (b2) at (3.05,1) {#6};
    \node[bit] (b3) at (3.05,0) {#7};

    \foreach \g/\b/\lab in {g1/b1/1, g2/b2/2, g3/b3/3}{%
      \draw[arr] (\g) -- node[above, font=\scriptsize, inner sep=1pt]
                 {$\operatorname{bit}_{\lab}$} (\b);}

    \coordinate (t) at (3.65,2);
    \coordinate (m) at (3.65,1);
    \coordinate (u) at (3.65,0);
    \draw[gray, semithick] (b1.east) -- (t) (b3.east) -- (u);
    \draw[gray, semithick] (t) -- (u);
    \node[doc] (d) at (4.85,1) {#9};
    \draw[emit] (m) -- (d);

    \node[font=\small] at (3.05,-0.8) {code #8};
  \end{tikzpicture}%
}

\title{Learning Query Encoders Can Be Hard Even When Vector Retrieval Is Geometrically Easy}

\author[1]{Anders Wikum \thanks{Work done while interning at Amazon AWS}}
\author[2]{Nina Mishra}
\author[1,2]{Amin Saberi}
\author[2, 3]{Tal Wagner}
\affil[1]{Stanford University}
\affil[2]{Amazon AWS}
\affil[3] {Tel Aviv University}
\date{}
\begin{document}
\maketitle

\begin{abstract}
Efficient vector retrieval requires both a corpus geometry that supports retrieving the right documents through vector similarity, and a query encoder that can embed queries near their desired documents in the embedding space. Recent work has studied geometric capacity through the lens of the minimum embedding dimension needed to realize all top-$k$ answer sets of $n$ documents. We study a different notion of geometric capacity—the maximum recall achievable for a frozen document index—and explore whether learned query encoders can reach this ceiling. On several real-world retrieval benchmarks, we show that retrieval quality of single-vector query encoders often lies far below what the document indices can support. 

Motivated by this observation, we give theoretical evidence that learning query encoders can be computationally hard. 
In particular, we construct a retrieval task that (1) admits a query encoder with perfect recall which is representable by a small one-hidden-layer ReLU network, but (2) any statistical-query learner (a class capturing learners that access training data through aggregate statistics) provably requires exponentially many statistical queries to achieve non-trivial recall
advantage over the random baseline $k/n$. Taken together, our results suggest substantial unrealized geometric capacity in retrieval benchmarks and establish query encoder learnability as a possible barrier in embedding-based retrieval.
\end{abstract}

\section{Introduction}

Efficient, large-scale modern vector retrieval systems are based on the bi-encoder architecture \citep{karpukhin20:dense,lewis20:retrieval} in which documents are embedded offline as vectors and stored in a \emph{frozen index} that can be reused across a stream of queries. At retrieval time, a \emph{query encoder} maps the incoming query into the shared vector space, and the system returns the $k$ best documents by similarity score (e.g. inner-product, cosine, or inverse $\ell_2$ distance). The success of this architecture depends on two distinct requirements:

\begin{CompactItemize}
\item the \textbf{embedded document geometry} must support accurate retrieval through vector similarity---that is, there must exist a query embedding whose $k$-nearest documents match the semantically intended documents for that query; and 
\item it must be \textbf{computationally feasible} to learn a query encoder that produces these query embeddings.
\end{CompactItemize}

The first has been the primary focus of recent work. Theoretically, it has been shown that even if documents are embedded optimally, geometry can be the limiting factor when embedding dimension is low and the retrieval task requires realizing many distinct top-$k$ subsets of the document corpus; at the same time, empirical studies have observed that these geometric limits do not necessarily take effect in practice \citep{weller26:theoretical,s26:strengths,wang26:theoretically,okajima26:limits,bangachev26:dimensionality}. These findings motivate the questions: to what extent does a document index support high recall, and subject to this ceiling, how much recall can a learned query encoder achieve? We investigate both questions empirically and theoretically. 

Our first contribution is an empirical study of the gap between the \emph{geometric capacity} of a document index (i.e. the best recall it can possibly support) and the recall  achieved by learned query encoders. We develop a ranking support vector machine (ranking SVM)-based method for efficiently gauging the geometric capacity of a given embedded corpus, then use it to study the recall attained on 20 datasets from the BRIGHT, BEIR and LIMIT retrieval benchmarks by various pre-trained and fine-tuned query encoders. 
Our results show that geometric capacity is rarely the limit---on most datasets, even a lightweight document encoder can support near-perfect retrieval. On the other hand, the tested query encoders---even when fine-tuned on queries and their gold sets---mostly fail to exceed 50\% or even 25\% recall on held-out test queries. These findings point to query-side encoding as a more salient barrier than geometric capacity in dense retrieval. These results are given in \Cref{sec:geom}.

Our second contribution is a theoretical result which shows the existence of a retrieval task for which learning a good query encoder is computationally hard. 
Namely, we construct a task for which a query encoder with perfect recall exists and can be represented as a small ReLU network, but learning a query encoder with better-than-random recall requires exponentially many steps in the statistical query (SQ) learning model \citep{kearns93:efficient}---a computational model for supervised learning in which a learner may only access the training set via aggregate statistics. This establishes a rigorous setting where retrieval quality is constrained not by the document geometry, but by the inability to efficiently learn a good query encoder.
This result is formalized and proved in \Cref{sec:barriers}. 

Our results complement the recent line of work on geometric limits of embedding-based retrieval, shining a spotlight on query encoding as a barrier and possible source of hardness.

\subsection{Related work}
\noindent\textbf{Geometric limits of vector retrieval.} A recent line of research has explored whether embedding dimension constrains the performance of single-vector retrieval methods by considering the question: \emph{what is the minimum dimension required to realize all top-$k$ subsets of $n$ documents?} \citet{weller26:theoretical} showed dimension lower bounds for realizability with and without margin, and introduced the LIMIT data to showcase these limitations empirically. Subsequent work has extended the no-margin lower bounds beyond inner product scoring \citep{wang26:theoretically}, characterized the minimum embedding dimension for max-margin separation \citep{bangachev26:dimensionality}, and examined further limits imposed by quantization \citep{okajima26:limits}. Each of these results assumes the ability to embed documents optimally. Empirically, independent observations have shown that geometry is not, in fact, the primary source of difficulty for the LIMIT dataset \citep{s26:strengths, wang26:theoretically}, instead attributing difficulty to mechanisms like drowning in documents or general learning failures. We show that this phenomenon is not unique to LIMIT---on realistic retrieval benchmarks, dense retrievers often achieve performance far below what is supported by the document geometry.
\\

\noindent\textbf{SQ lower bounds.} The statistical query (SQ) model \citep{kearns93:efficient} restricts a learner's sample access to aggregate statistics of the training distribution; a setting that captures gradient-based learning under suitable precision and batch-size conditions \citep{abbe21:power}. Following the foundational result that parity functions are not efficiently SQ-learnable \citep{kearns93:efficient}, subsequent work has characterized weak SQ-learnability for function classes \citep{blum94:weakly} and developed SQ lower bounds for search problems in which an algorithm must produce an output satisfying a success criterion \citep{feldman13:statistical}. These techniques have yielded lower bounds for learning shallow neural networks \citep{song17:complexity,goel20:superpolynomial,chen22:hardness}, robust estimation of gaussians \citep{diakonikolas17:statistical}, and robust classification \citep{bubeck19:adversarial}. In this work, we use these techniques to show that learning a single-vector retrieval model with nontrivial recall can be SQ-hard.

\section{Preliminaries}
\noindent\textbf{Information retrieval.}
A retrieval task consists of a corpus of $n$ documents $d_1, \dots, d_n$ and a
distribution $\mathcal{Q}$ over natural language queries. Each query $q$ has an
associated set of retrieval targets $G(q) \subset \{d_1, \dots, d_n\}$. We refer
to $G(q)$ as the \emph{gold documents} for query $q$. The goal is to retrieve a set of documents $S$ with high recall, where recall is defined as
\[\rec(S, G(q)) = \frac{|S \cap G(q)|}{|G(q)|}.\]
A retrieval rule $R: q \mapsto S(q)$ specifies the retrieved set for every query $q$. For a distribution over queries $\mathcal{Q}$, we are interested in learning a retrieval rule that maximizes the expected recall of the returned documents
\[\E_{q \sim \mathcal{Q}}[\rec(R(q), G(q))].\]

\noindent\textbf{Single-vector retrieval.} The dominant paradigm for single-vector retrieval is a bi-encoder architecture, in which document embeddings $x_{d_i} \in \R^d$ are computed in advance and stored in a \emph{frozen index} $\mathcal{X}$, while a separate query encoder maps queries into the shared vector space at inference time for retrieval. In this setting, the query encoder induces a retrieval rule by mapping the query $q$ to a vector $E(q) \in \R^d$ then returning the $k$ documents with largest inner product score $s(q, d_i)=\langle E(q), x_{d_i}\rangle$:
\[
  R_E^{(k)}(q) := \Topk\bigl(E(q); \mathcal{X}\bigr).
\]
We write $R_E(q) = R_E^{|G(q)|}(q)$ for brevity, and when the underlying query encoder $E$ is obvious from context, we denote the expected recall of the single-vector retriever as
\[r_{\operatorname{encoder}}(\mathcal{Q}; \mathcal{X}) := \E_{q \sim \mathcal{Q}}\operatorname{recall}(R_E(q), G(q))].\]

\noindent\textbf{The SQ learning model.}
Introduced by \citet{kearns93:efficient}, in the SQ model we seek to learn a function $f: x \mapsto y$ given sample access to an unknown distribution $\Pi$ over labeled training examples $(x,y)$ only through statistical queries. Statistical queries differ from arbitrary sample access in that they can only output aggregate statistics of the training data.

\begin{definition}\label{def:sq}
Let $\Pi$ be a distribution over training samples and let $\tau > 0$ be a fixed \emph{tolerance} parameter. The \emph{statistical query} oracle $\STAT_\Pi(\tau)$ accepts any query function
$\varphi : (x,y) \mapsto [-1,1]$ and returns some value $v \in \EE_\Pi[\varphi] \pm \tau$.
\end{definition}
The oracle may be adversarial within the tolerance $\tau$. Lower bounds in this model apply to any algorithm implementable by an SQ oracle, including method-of-moments estimators, spectral
methods, and mini-batch SGD under precision and batch-size conditions
\citep{feldman17:characterization, abbe21:power}.\\

\section{Beyond Geometric Limitations}\label{sec:geom}
Prior work on geometric limitations of embedding-based retrieval has focused on characterizing the minimal embedding dimension in which, under perfect document embedding, all document subsets of size $k$ can be realized. In practice, however, documents are embedded imperfectly by a semantic encoder, and there are only a small number of gold sets to retrieve. These observations motivate a dual notion of geometric capacity: for a frozen index $\mathcal{X}$ and query distribution $\mathcal{Q}$, what is the best possible recall $c_{\operatorname{geom}}(\mathcal{Q}; \mathcal{X})$ achievable by a single-vector query encoder? Formally, we write

\begin{equation}\label{eq:geom-cap}
c_{\operatorname{geom}}(\mathcal{Q}; \mathcal{X}):=\underset{q \sim \mathcal{Q}}{\E}[\sup_{E}\rec(R_E(q), G(q))].
\end{equation}

The goal of this section is to characterize the gap between geometric capacity and the recall achieved by learned query encoders on real-world retrieval benchmarks. To this end, the section is organized as follows. In \Cref{sec:ranking}, we describe how to efficiently compute a capacity lower bound $\hat{c}_{\operatorname{geom}}(\mathcal{Q}; \mathcal{X})\leq c_{\operatorname{geom}}(\mathcal{Q}; \mathcal{X})$ via a direct optimization approach.
Then, in \Cref{sec:experiments}, we apply this approach to a number of real-world retrieval benchmarks. We then evaluate the achieved recall $r_{\operatorname{encoder}}(\mathcal{Q}; \mathcal{X})$ of various pre-trained and fine-tuned semantic query encoders on these benchmakrs. 
If $r_{\operatorname{encoder}}(\mathcal{Q}; \mathcal{X})\ll \hat{c}_{\operatorname{geom}}\leq c_{\operatorname{geom}}(\mathcal{Q}; \mathcal{X})$, the tested query encoders fail to attain the recall enabled by the geometric capacity of the benchmark. 

\subsection{Capacity lower bounds via ranking SVM}\label{sec:ranking}
The purpose of this section is to provide a practical method for computing the lower bounds on the geometric capacity $c_{\operatorname{geom}}(\mathcal{Q}; \mathcal{X})$ of a given retrieval benchmark $(\mathcal{Q}; \mathcal{X})$.
To motivate our approach, first notice that the gold set of documents $G(q)$ for query $q$ is perfectly retrievable by a single-vector query embedder precisely when some query embedding $E(q)$ ranks every gold document above every non-gold document by inner product score. That is,  $\langle E(q), x_p-x_n\rangle > 0$ for each gold document $p \in G(q)$ and each non-gold document $n \notin G(q)$. Thus, finding such a query embedding amounts to finding a separating hyperplane for the pair-wise difference vectors $\{x_p - x_n\}_{p \in G(q), n \notin G(q)}$. This is exactly the convex feasibility problem underlying ranking SVM \citep{herbrich99:large,joachims06:training}, an early learned document ranking approach that attracted significant attention in the 2000s.

Formally, let $\mathcal{C} = G(q) \times G(q)^c$ index pairs of gold and non-gold documents for query $q$. For a fixed penalty $\alpha > 0$, the soft-margin ranking SVM constructs a query vector by minimizing a regularized pairwise ranking loss:
\vskip -2em
\begin{align}\label{eq:svm}
    \min\limits_{x_q \in \mathbb{R}^d, \; \xi \succeq 0} \quad &\frac{1}{2}\|x_q\|^2 + \alpha \sum_{(p,n) \in \mathcal{C}} \xi_{p,n} \\
    \text{s.t.} \quad &\langle x_q, x_p - x_n \rangle \geq 1 - \xi_{p,n}, \quad \forall (p,n) \in \mathcal{C} \notag
\end{align}

The constraints ensure that all documents are properly ranked up to the slack terms $\xi_{p,n}$. 
A feasible zero-slack solution ranks every gold document above every non-gold document and therefore certifies that the gold set $G(q)$ is perfectly retrievable under the document geometry. Furthermore, in any feasible solution (with or without slack), the recall attained by $x_q$ is a lower bound on the best recall possible for $q$ under the document geometry. 
Thus, we compute the lower bound $\hat{c}_{\operatorname{geom}}(\mathcal{Q}; \mathcal{X})$ on \eqref{eq:geom-cap} by solving the ranking SVM for every query in a dataset, and averaging the recall of highest-scoring $|G(q)|$ documents under the produced query embeddings $x_q^*$.

A practical challenge is that the ranking SVM formulation can have a huge number of pair-wise ranking constraints. We therefore solve \eqref{eq:svm} using the cutting-plane method introduced by \citet{joachims06:training}, which avoids explicit storage of all pairwise constraints. In practice, this allows us to directly optimize query embeddings for retrieval benchmarks with up to $\sim$500k documents within 2-4 hours. For each query, we grid search over the penalty $\alpha = \frac{C}{n}$ for $C \in \{10, 10^2, \dots, 10^6\}$.

\subsection{Experiments}\label{sec:experiments}
In this section, we explore the gap between geometric capacity and the recall achieved by learned query encoders on real-world retrieval datasets. \\

\noindent\textbf{Datasets.} We evaluate on three popular information retrieval benchmarks. Additional details on the composition of each dataset can be found in Appendix~\ref{app:exp-setup}.

\begin{CompactItemize}
  \item LIMIT \citep{weller26:theoretical} is a synthetic dataset designed to exhibit limitations of embedding-based retrieval. Each query names an attribute that appears in exactly two documents in a corpus of size 50k. Gold sets cover all top-2 subsets of 46 documents, a construction in which no sufficiently low-dimensional embedding model can achieve perfect recall.
  \item BRIGHT \citep{su25:bright} is a reasoning-intensive QA benchmark that spans 12 categories, 
  each ranging from 7.9k to 414k documents. Queries are complex questions, and document relevance is determined by the extent to which it helps answer the question. Unlike many other retrieval settings, in BRIGHT gold documents often share little keyword overlap with the query. 
  \item BEIR \citep{thakur21:beir} is a diverse suite of retrieval tasks with corpora ranging from 3.6k to millions of documents. Though originally intended to evaluate zero-shot generalization, the diverse categories are highly useful for our in-domain experiments. Due to compute and licensing constraints, we restrict our evaluation to 7 categories listed in the appendix. \\
\end{CompactItemize}

\noindent\textbf{Models.} We evaluate two single-vector embedding models: \textsf{DistilBERT} \citep{sanh20:distilbert}, a lightweight BERT encoder with embedding dimension $768$, and \textsf{Qwen3-Embedding-4B} \citep{zhang25:qwen3}, a large semantic encoder with embedding dimension $2560$. Following the setup of \citet{weller26:theoretical}, as a point of comparison we also measure recall achieved by models with a different inductive bias: lexical search via \textsf{BM25} \citep{robertson2009probabilistic}, and multi-vector retrieval with late-interaction scoring via \textsf{GTE-ModernColBERT} \citep{khattab20:colbert,ModernColBERT}. We note that geometric capacity is only a ceiling for single-vector retrievers.\\

\noindent\textbf{Fine-tuning.}
When train-test splits are not pre-specified, we partition queries into a 70/30 train-test split. The three semantic embedding models are each fine-tuned using LoRA \citep{hu22:lora} (rank 16) with AdamW optimizer \citep{loshchilov17:decoupled} and a contrastive multi-positive InfoNCE loss for 3,000 steps at the default learning rate. Negatives are sampled for corpora exceeding $200$k documents (100,000) and for multi-vector training (512). Due to limited data for validation, we report the model checkpoint at step 3,000.

\subsubsection{Results}
For full numerical results, see Appendix~\ref{app:results}.\\\\
\noindent\textbf{Geometric capacity is not limiting.} \Cref{fig:gap} shows the gap between the expected recall achieved by (1) pre-trained and fine-tuned single-vector embedding models ($r_{\operatorname{encoder}}(\mathcal{Q}; \mathcal{X})$) and (2) the directly-optimized geometric capacity lower bound ($\hat{c}_{\operatorname{geom}}(\mathcal{Q}; \mathcal{X}) \leq c_{\operatorname{geom}}(\mathcal{Q}; \mathcal{X})$) across our benchmark datasets. The results are striking; outside of BEIR's \textsf{trec-covid} (which is an outlier with median gold sets of size 529), we see geometric capacity of 95\%+ expected recall across nearly all retrieval tasks and both document embedding models. Notably, even the lightweight semantic encoder DistilBERT is able to embed documents in such a way that near-perfect retrieval is (geometrically) possible. On the flip side, pre-trained and fine-tuned query encoders rarely exceed 50\% average recall and often perform less than 25\% on the evaluation queries. These results are 
consistent with earlier empirical observations that fine-tuned single-vector models struggle to generalize on LIMIT-like datasets \citep{s26:strengths}.

\begin{figure}
\centering
  \includegraphics[width=0.9\textwidth]{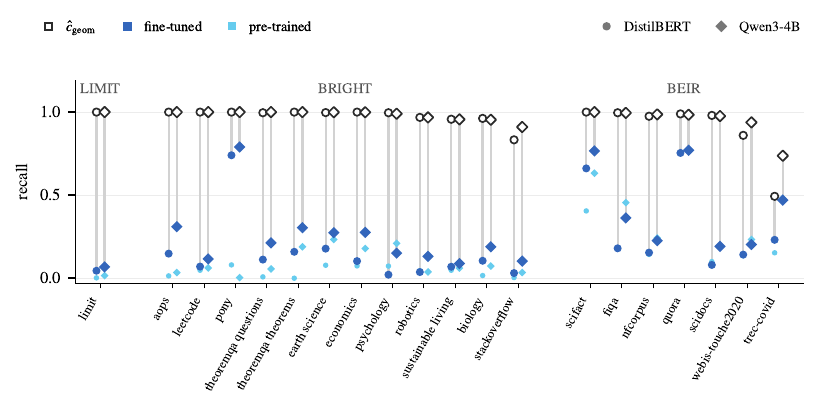}
  \vskip -1em
  \caption{Measured gap between geometric capacity of the document index and the expected recall of single-vector query encoders.}
  \label{fig:gap}
\end{figure}

\begin{figure}
  \centering
  \includegraphics[width=0.9\textwidth]{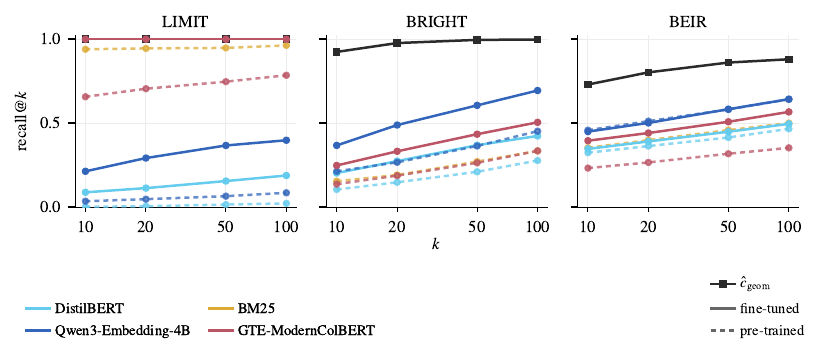}
  \vskip -1em
  \caption{Expected recall@$k$ for $k \in \{10, 20, 50, 100\}$ averaged
  over the datasets in each benchmark. On
  LIMIT, lexical and fine-tuned ColBERT-style multi-vector retrieval are near-saturated at every
  cutoff while the single-vector encoders remain far below even at $k=100$.}
  \label{fig:k-scale}
\end{figure}

While \Cref{fig:gap} shows results for only the top $|G(q)|$ highest-scoring documents, we note that over-retrieving a larger number of candidate documents than needed and reranking them is common in modern retrieval pipelines. \Cref{fig:k-scale} shows how performance scales as the size of the returned set of documents $k$ is allowed to increase. Although retrieval quality increases with the number of documents retrieved across all benchmarks, expected recall lies below the geometric ceiling lower bound for single-vector models even with $k=100$ returned documents.

\noindent\textbf{Inductive bias and generalization.} To better understand the poor performance of single-vector query encoders, in \Cref{fig:train-test} we show expected recall achieved on train/test splits across each dataset. We find that fine-tuning generally fits the training queries but does not generalize. On BRIGHT, fine-tuned \textsf{DistilBERT} and \textsf{Qwen3-Embedding-4B} models achieve an average train recall of $0.957$ and $0.982$, respectively. At the same time, average recall over held-out queries in the same categories is just $0.148$ and $0.245$---a clear signal that LoRA is failing to learn a generalizable single-vector mapping between query and documents. 

LIMIT cleanly isolates the role that inductive bias can play in generalization. Each of its $1000$ queries has a gold set drawn from shared pool of $46$ documents, so every gold document is seen during training. Still,
\textsf{Qwen3-Embedding-4B} attains recall of just
$0.068$ on held-out queries despite perfectly learning the training set. In contrast, a fine-tuned multi-vector retrieval model \textsf{GTE-ModernColBERT} attains train/test recall of $1.000/0.998$. This may be explained, in part, by the fact that late interaction scoring provides a direct mechanism for matching query terms to document terms \citep{khattab20:colbert}.

\begin{figure}
  \centering

\includegraphics[width=0.9\textwidth]{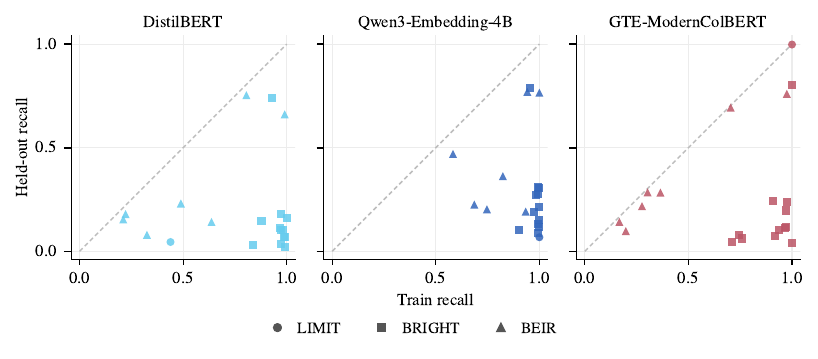}
\vskip -1em
  \caption{Expected recall on training vs. held-out queries for each fine-tuned retrieval model. Points in the lower right reflect near-perfect recall on training data but limited generalization.}
  \label{fig:train-test}
\end{figure}

\section{Computational Barriers in the SQ Model} In \label{sec:barriers}
\Cref{sec:geom}, we observed that learned query encoders fail to exploit the geometric capacity of the document index on several real-world retrieval benchmarks. 
This gap may arise for several reasons:
\begin{CompactItemize}
    \item Perhaps, even though good query embeddings exist geometrically, they cannot be realized by a succinct query embedding model.
    \item Perhaps they can be realized by a succinct model, but we haven't yet found successful learning algorithms for training it.
    \item Perhaps they can be realized by a succinct model, but learning it is computationally hard. 
\end{CompactItemize}

In this section, we give theoretical evidence that the third bullet---computational hardness for learning query embeddings---may occur. 
In particular, we give a construction in which perfect recall can be achieved by a small one-hidden-layer ReLU neural network (\Cref{lemma:realize}), but any learner in the SQ model requires exponential computation time to produce an encoder with better-than-random recall (\Cref{thm:main}). This shows that, under SQ learners, high geometric capacity does not mean the certifying query embedding is actually learnable. Detailed proofs can be found in Appendix~\ref{app:sq}.


\subsection{Retrieval construction}

\begin{figure}
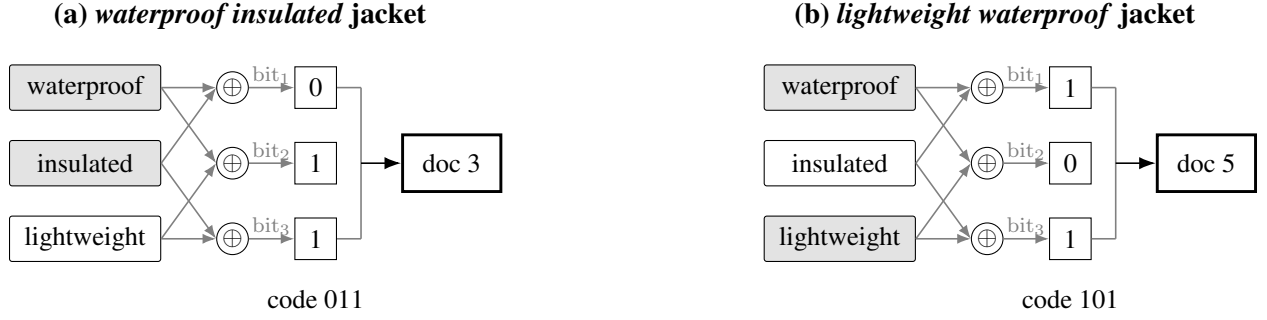

\centering
\xorpanel{(a) \emph{waterproof insulated} jacket}%
  {gray!22}{gray!22}{white}{0}{1}{1}{011}{doc 3}%
\hfill%
\xorpanel{(b) \emph{lightweight waterproof} jacket}%
  {gray!22}{white}{gray!22}{1}{0}{1}{101}{doc 5}
\caption{A simple intent. Each document code bit is the parity of a fixed
pair of keywords. The two queries share \emph{waterproof} and only differ in a
single term, yet the resulting retrieval codes are dramatically different.}
\label{fig:xor}
\end{figure}

At a high level, we construct a retrieval task in which queries may be interpreted in many ways, called \emph{intents}, each of which produces a different gold set of documents. A learner is given access to a labeled training set of (query, gold set) pairs generated by a hidden intent. Computational hardness will stem from making samples from different intents look indistinguishable under aggregate statistics.\\

\noindent\textbf{Notation.} 
Our analysis will involve arithmetic operations both over the binary finite field $\FF_2$ and over $\mathbb{R}$. For clarity, we use ``$\oplus$'' for binary addition (xor) over $\FF_2$. Binary matrix-vector multiplication $Aq$ is defined by $(Aq)_i=\oplus_jA_{ij}q_j$. 
We further use $|b|$ to denote the number of $1$'s in a binary vector $b$. 
Embedding vectors will lie in a vector space over $\R$, and their inner products and the sums defining retrieval scores use ordinary real-valued arithmetic. \\

\noindent\textbf{Query model.}
We model a natural language query as a binary string of length $P$  which may represent, for example, a bag-of-words indicator vector over a vocabulary of $P$ keywords. Queries are drawn uniformly at random from $\FF_2^P$. That is, the query distribution is 
$\mathcal{Q} = \Unif(\FF_2^P)$.\\

\noindent\textbf{Document corpus.} Let $d = \Theta(\log P)$. The corpus consists of $n=2^d$ documents, each assigned a unique binary code $w \in \FF^d_2$ and embedded as the hypercube vertex \[h(w) = \bigl( (-1)^{w_1}, \dots, (-1)^{w_d} \bigr)^{\!\top} \in \{-1,1\}^d.\]
An \emph{intent} $A \in \FF^{d \times P}_2$ maps each query $q$ to a document code $Aq$ which serves as the center of the query's gold set. This matrix-vector product is over $\FF_2$, so the $i^{th}$ bit of the binary string $Aq$ is the parity of keywords $j$ in the query for which $A_{ij} = 1$. See \Cref{fig:xor} for a simple instantiation of $A$ and $q$.

Fix an integer radius $0\leq \rho< d$ independent of $P$. The gold set for query $q$ under intent $A$ consists of documents with binary codes within hamming distance $\rho$ from $Aq$, of which there are
$k = \sum_{j=0}^{\rho}\binom{d}{j} = (\log P)^{O(\rho)}$.\\

\noindent\textbf{Training distribution.} Each training sample $(q, G_A(q))$ consists of a query $q \sim \mathcal{Q}$ and
its full gold document set. Note that perfect knowledge of the gold set is a strong
supervision signal that requires relevance judgments over the entire corpus of
documents. Proving computational hardness is strictly harder in this regime than
under partial feedback produced by, e.g., clicks and user interactions. We write
$D_A$ for the distribution over training pairs induced by the intent $A$.

\subsection{Perfect geometric realizability}
First, we show that the document embeddings on the hypercube defined above support a single-vector query encoder that achieves perfect recall.
\label{sec:geometry}
\begin{restatable}{lemma}{PerfRealize}\label{lemma:realize}
  For every intent $A$, there exists a query encoder with perfect top-$k$ recall, which is representable as a polynomial-size, one-hidden-layer ReLU network.
\end{restatable}
\begin{proof}[Proof sketch] Consider the query encoder $E: q \mapsto h(Aq)$ for the hidden intent $A$. For query $q$, the score assigned to document $w$ is \[s(w) = \langle h(Aq), h(w)\rangle = \sum_{i=1}^d \left(\textbf{1}\{(Aq)_i = w_i\} - \textbf{1}\{(Aq)_i \neq w_i\}\right) = d - 2 d_H(Aq, w).\]
Thus, the top-$k$ retrieval set is the $k$ closest documents to $Aq$ in Hamming distance. The proof of realizability by a shallow neural network follows standard arguments for representing parity as a ReLU spline and can be found in \Cref{app:realize}.
\end{proof}

\subsection{SQ-hardness}
\Cref{lemma:realize} shows that a simple, high-recall query encoder exists that may act as a learning target. Nonetheless, our main result shows that learning this encoder---or in fact any retrieval rule with better-than-random recall---can be computationally hard.
Throughout this section, every retrieval rule returns exactly \(k\) documents, and its recall is the expectation over a fresh query.

\begin{theorem}\label{thm:main}
There is a family of intents $\mathcal{A}$ such that, for any
inverse-polynomial tolerance $\tau$ and inverse-polynomial target advantage
$\varepsilon$, every SQ learner that achieves recall at least $k/n + \varepsilon$
for every intent with constant success probability requires $2^{\Omega(P)}$
statistical queries.
\end{theorem}

At a high-level, the idea is to construct a large family of intents $\mathcal{A}$ in which (1) achieving a non-trivial advantage over random $k/n$ recall requires identifying the hidden intent $A \in \mathcal{A}$ up to a small set of candidates, and (2) samples from all intents look identical under statistical queries. That is, any two intents map to a shared center at the same rate as a pure chance:
\begin{equation}\label{eq:chance-collision}
  \Pr_{q \sim \mathcal{Q}}[Aq = A'q] = \frac{1}{n} \quad \text{for all $A \neq A' \in \mathcal{A}.$}
\end{equation}

\noindent\textbf{Intent family.}
We show the existence of an intent family satisfying  \eqref{eq:chance-collision} 
using the probabilistic method. Consider the following randomized construction: draw $N = \bigl\lfloor \frac{2^{P/3}}{n} \bigr\rfloor = 2^{\Omega(P)}$ binary matrices $\mathcal{A} = \{A^{(1)}, \dots, A^{(N)}\}$ independently and uniformly from $\FF_2^{d \times P}$. 
Given a query $q\in\FF_2^P$, observe that two intents $A^{(i)}, A^{(j)}$ map to the same center ($A^{(i)}q=A^{(j)}q$) precisely when $(A^{(i)} \oplus A^{(j)})q = 0$, that is, when $q$ lies in $\ker(A^{(i)} \oplus A^{(j)})$, where $\ker(A)=\{v\in\FF_2^P:Av=0\}$. By the rank nullity theorem for $\FF_2^P$, $|\ker(A)|=2^{P-\operatorname{rank}(A)}$ for every $A$ where the matrix rank is over $\FF_2$. Hence,
a uniform $q$ falls in $\ker(A^{(i)} \oplus A^{(j)})$ with probability $2^{-P}|\ker(A^{(i)} \oplus A^{(j)})|
  = 2^{-\operatorname{rank}(A^{(i)} \oplus A^{(j)})}$,

Thus, $\mathcal{A}$ consists of $N = 2^{\Omega(P)}$ distinct intents satisfying
\eqref{eq:chance-collision} precisely when each $\FF_2$-matrix sum $A^{(i)} \oplus A^{(j)}$ has rank $d = \log n$. The following lemma shows that this occurs with positive probability.

\begin{restatable}{lemma}{lemFamily}\label{lem:family}
With probability at least $1/2$, every pairwise difference
has $\operatorname{rank}(A^{(i)} \oplus A^{(j)}) = d$.
\end{restatable}
\begin{proof}[Proof sketch]
$A^{(i)} \oplus A^{(j)}$ is a uniform binary matrix. Building such a matrix by sampling row-by-row, the $\ell^{th}$ row lies in the span of the first $\ell - 1$ with probability at most $2^{\ell-1}/2^P$ since the span of $\ell-1$ binary vectors contains at most $2^{\ell-1}$ binary vectors. 
 A union bound over $d$ rows and the $\binom{N}{2}$ difference matrices completes the proof; see \Cref{app:sq} for details. 
\end{proof}

\noindent\textbf{SQ lower bound.} With our exponential-size family of intents $\mathcal{A}$ in hand, we are now ready to state our main computational lower bound. 
\begin{lemma}\label{lem:sqbarrier}
Let $\mathcal{A}$ be a family of intents satisfying condition \eqref{eq:chance-collision}. Any randomized learner that outputs, for every intent $A \in \mathcal{A}$, a retrieval rule achieving recall at least $k/n + \varepsilon$ with probability at least $\delta$ requires $T_{\operatorname{SQ}}$ statistical queries at tolerance $\tau$, where
\begin{equation}\label{eq:sq}
  T_{\operatorname{SQ}} \;\ge\; \Bigl( \delta|\mathcal{A}| - \tfrac{k}{\varepsilon^2} \Bigr)
  \cdot \frac{\tau^2}{n-1} .
\end{equation}
\end{lemma}

We observe that \Cref{lem:sqbarrier} implies \Cref{thm:main} for our intent family $\mathcal{A}$:

\begin{proof}[Proof of Theorem~\ref{thm:main}] 
Recall that $n=2^d\leq P^c$ for some constant $c>0$. 
For inverse-polynomial $\tau$ and $\varepsilon$ we have $\tau\geq P^{-c'}$,  $\varepsilon\geq P^{-c''}$ for some constants $c',c''>0$. 
By construction, $|\mathcal{A}| = N =  \lfloor 2^{P/3}/n\rfloor \geq 2^{P/3}/(2P^c)$. 
Since $\delta$ is constant and $k=(\log P)^{O(\rho)}$ with $\rho$ independent of $P$, we have $\delta|\mathcal{A}|\geq 2k/\varepsilon^2$ for sufficiently large $P$. Plugging these in \eqref{eq:sq} yields $T_{\mathrm{SQ}} \geq 2^{\Omega(P)}$.
\end{proof}

The remainder of this section is dedicated to the proof of \Cref{lem:sqbarrier}. Our road map is as follows. First, we show that a fixed retrieval rule can only achieve recall $k/n + \varepsilon$ on a small fraction of intents $\mathcal{A}$. Thus, beating baseline recall requires identifying the hidden intent $A \in \mathcal{A}$ up to a small candidate set. The bound then follows from counting how many intents each statistical query rules out.

To begin, it will be useful to define the distribution $D_A$ over $(q, Aq)$ pairs induced by intent $A$, and the null distribution $\mathcal{N}_0$ in which the query $q$ and gold set center $Aq$ are independent:
  \[D_A(q, y) = \frac{1}{2^P} \cdot \textbf{1}\{y = Aq\}, \quad\quad  \mathcal{N}_0(q, y) = \frac{1}{n\cdot 2^P}. \]

  \Cref{lemma:measure-gap} shows that any statistical query distinguishes $D_A$ from the null $\mathcal{N}_0$ at tolerance $\tau$ for only a small number of intents $A \in \mathcal{A}$. The proof is technical and left to Appendix~\ref{app:sq}.
  \begin{restatable}{lemma}{measureGap}\label{lemma:measure-gap}
    For any statistical query $\varphi: \FF^P_2 \times \FF^d _2 \to [-1, 1]$, define the measurement gap $\Delta_A(\varphi) := \E_{D_A}[\varphi] - \E_{\mathcal{N}_0}[\varphi]$ under intent $A$. Then, $|\Delta_A(\varphi)| > \tau$ for at most $\frac{(n-1)\cdot\operatorname{Var}_{\mathcal{N}_0}(\varphi)}{\tau^2}$ intents.
  \end{restatable}
  Treating recall as a statistical query, an immediate consequence is that no fixed retrieval rule can simultaneously achieve better-than-random recall for many intents.
\begin{restatable}{corollary}{intentCov}
  \label{lem:intent-cov}
  Any fixed retrieval rule achieves expected recall at least $\frac{k}{n} + \varepsilon$ for at most $\frac{k}{\varepsilon^2}$ intents.
\end{restatable}
\begin{proof}[Proof sketch]
The statistical query $\varphi_R(q,y) \;=\; \rec\bigl(R(q),\, G(y)\bigr) \in [0,1]$ measures recall against a gold set $G(y)$ centered at $y$. Under $D_A$, $y = Aq$ and $\E_{D_A}[\varphi_R]
= \E_{q \sim \mathcal{Q}}\bigl[ \rec(R(q), G_A(q)) \bigr].$ Under $\mathcal{N}_0$, the center $y$ is
uniform and independent of $q$, so $\E_{\mathcal{N}_0}[\varphi_R] = k/n$. Thus, the measurement gap is $\Delta_A(\varphi_R) = \E_{D_A}[\rec(R(q), G_A(q))] - \frac{k}{n}$. Moreover, $\varphi_R^2 \le \varphi_R$ since $\varphi_R \in [0,1]$, and so
\[
  \Var_{\mathcal{N}_0}(\varphi_R)
  \;\le\; \E_{\mathcal{N}_0}[\varphi_R] - \Bigl(\frac{k}{n}\Bigr)^{\!2}
  \;=\; \frac{k}{n}\Bigl(1 - \frac{k}{n}\Bigr) \leq \frac{k}{n}.
\]
An application of \Cref{lemma:measure-gap} with $\tau = \varepsilon$ yields the result.
\end{proof}

With these results, we now give the proof of \Cref{lem:sqbarrier}.
\begin{proof}[Proof of \Cref{lem:sqbarrier}]  For each $A \in \mathcal{A}$, consider the oracle \[O_A(\varphi) =
  \begin{cases}
    \EE_{\mathcal{N}_0}[\varphi] &\text{if } |\Delta_A(\varphi)| \le \tau \\
    \EE_{D_A}[\varphi] & \text{otherwise}.
  \end{cases}.\]
The SQ learner must succeed against every $\STAT_{D_A}(\tau)$ oracle, hence also against each $\{O_A\}_{A \in \mathcal{A}}$. Fix the random seed $s$ of the SQ learner and let $W_s \subseteq \mathcal{A}$ be the set of intents for which the learner succeeds, i.e., produces a retrieval rule whose recall is at least $k/n + \varepsilon$, given access to the oracle $O_A$. As a point of reference, we will consider a hypothetical sequence of queries $\varphi_1, \dots, \varphi_{T_{\operatorname{SQ}}}$ produced by the learner with access to an oracle that only answers with the null value $\EE_{\mathcal{N}_0}[\varphi]$. For any intent $A \in W_s$ in which the learner succeeds, either (i) $O_A$ produces the same answers as the reference oracle, resulting in the same retrieval rule $R$, or (ii) $O_A$'s answers differ from the reference oracle at some timestep $t \leq T_{\operatorname{SQ}}$. Thus, we can write
\[|W_s| \leq  |\{A \in \mathcal{A}: R \text{\;achieves recall\;} k/n + \varepsilon\}| + |\{A \in \mathcal{A}: |\Delta_A(\varphi_t)| > \tau \text{\;for some\;} t \leq T_{\operatorname{SQ}}\}|.\]
The first set has size at most $k/\varepsilon^2$ by \Cref{lem:intent-cov}. For the second, define $S_t := \{A \in \mathcal{A}: |\Delta_A(\varphi_t)| > \tau\}$. By \Cref{lemma:measure-gap} and the fact that a random variable taking values in $[-1,1]$ has variance at most 1, $|S_t| < (n-1)/\tau^2$. Summing over the $T_{\operatorname{SQ}}$ queries made by the SQ learner yields that the second set has size at most $\tau^{-2}T_{\operatorname{SQ}}\,(n-1)$. 

Finally, averaging over the randomness of the learner, which has success probability at least $\delta$ for every $A\in \mathcal{A}$ and every $\operatorname{STAT}_{D_A}(\tau)$ oracle, yields the desired result up to rearranging terms:
$\delta\,|\mathcal{A}| \leq \E[|W_s|]
  \;\le\; \frac{T_{\operatorname{SQ}}\,(n-1)}{\tau^2} + \frac{k}{\varepsilon^2}$. 
\end{proof}

\section{Conclusion}
Complementing recent work on the geometric limitations of embedding-based retrieval, we explore whether learned query encoders can exploit the geometric capacity of a frozen document index. Our main takeaway is that producing a high-quality query encoder can pose a substantial barrier even when the geometry of the underlying document index is not restrictive. We exhibit this empirically on real-world encoder models and benchmarks through a ranking-SVM method to gauge geometric capacity, and theoretically through a computationally hard-to-learn family of retrieval tasks in the statistical query model. See \Cref{app:limitations} for a discussion of our study's limitations.  
\addcontentsline{toc}{section}{Conclusion}

\newpage

\subsection*{AI use statement}
During the work on this research we used Anthropic's Opus 5 and OpenAI's GPT 5.6 Sol to aid in interactive coding, literature search, proof verification and simplification, and writing adjustments and polishing. The authors are accountable for the manuscript, proofs and code in their entirety.

\bibliographystyle{plainnat}
\bibliography{refs}

\appendix

\section{Limitations}\label{app:limitations}
\noindent\textbf{Empirical results.}
The empirical gap we demonstrate on real-world encoders and benchmarks is, in some part, an artifact of our specific fine-tuning setup; it does not preclude the existence of better query encoders which might have not been tested or not yet discovered. 
Furthermore, although we include multi-vector models as a point of reference, our ranking SVM bounds on geometric capacity do not immediately extend to the multi-vector setting, and we leave it to future work to directly optimize multi-vector embeddings. \\

\noindent\textbf{Theoretical results.}
Our computational hardness result is proven in the SQ model and thus limited to learning algorithms that adhere to it, and our hard retrieval task construction is necessarily contrived for the sake of the proof; it may not resemble retrieval tasks that occur in real applications. Overall, while our results are highly suggestive that query encoding may be a substantial barrier to better recall in dense retrieval, future work may develop better query encoding methods and evolve our understanding of this field. 

\section{Omitted Proofs}
\subsection{Perfect realizability} \label{app:realize}
This section shows that the embedding model $q \mapsto Aq$ can be represented by a simple ReLU network with one hidden layer. We first represent parity by a one-dimensional ReLU spline.
\begin{lemma}\label{lem:parity}
  Let $\sigma(t) = \max\{t, 0\}$.
For all integers $t \in \{0,\dots,P\}$, the following ReLU spline satisfies $\ p_P(t) = (-1)^t$:
\[p_P(t) := 1 - 2\sigma(t) + 4\sum_{i=1}^{P-1} (-1)^{i+1}\,\sigma(t-i).\]

\end{lemma}

\begin{proof}
We have $p_P(0) = 1$. Note that over the interval $t \in (j, j+1)$, $p_P(t)$ is a linear function with slope
\[\frac{d}{dt} \; p_P(t) = -2 + 4\sum_{i=1}^j (-1)^{i+1} = \begin{cases}
2 &\text{if $j$ odd} \\
-2 &\text{if $j$ even}
\end{cases}.\]
Thus, the value of $p_P(t)$ alternates between $1$ and $-1$ at consecutive integers.
\end{proof}

For the binary-cube index, each coordinate of $h(Aq)$ is one parity of
keywords in $q$, so Lemma~\ref{lem:parity} gives width $dP$ directly. 
\subsection{Recovery from individual samples}
This appendix proves the upper bound accompanying Theorem~\ref{thm:main}.

\begin{lemma}
  Given access to individual training samples $(q, G_A(q))$ from $D_A$, it is possible to recover the hidden intent $A$ in $O(P^3)$ time with high probability.
\end{lemma}
\begin{proof}
  First note that by construction, the gold set $G_A(q)$ uniquely identifies its center $Aq$. Thus, our task amounts to reconstructing a binary matrix $A$ from samples of matrix-vector products $(q_t, Aq_t)_{t=1}^T$. In particular, row $j$ of $A$ is determined by the linear equations
  \[A_j q_t = (Aq_t)_j \quad \quad \text{for all $t=1, \dots, T$}.\]
  This system has a unique solution once $\operatorname{span}(q_1, \dots, q_T) = \FF_2^P$. For any nonzero $u \in \FF_2^P$ we have $\Pr[\langle q_t, u\rangle = 0 \text{ for all } t=1, \dots, T] = 2^{-T}$ since $\langle q_t, u\rangle$ is a uniform random bit. Therefore, a union bound over the $2^P - 1$ nonzero $u$ shows that $\{q_1, \dots, q_T\}$ fails to span $\FF^P_2$ with probability at most $2^{P-T}$. Taking $T = 2P$, Gaussian elimination
recovers all $d$ rows of $A$ in time $O(P^3)$ with probability at least
$1 - 2^{-P}$.
\end{proof}

Combined with Lemma~\ref{lemma:realize}, recovering $A$ yields an encoder of recall one, so intents are efficiently learnable from individual samples and the computational barrier in Theorem~\ref{thm:main} is specific to statistical access.

\subsection{SQ Hardness}\label{app:sq}

\lemFamily*

\begin{proof}
Fix $i \neq j$. The xor-sum $A^{(i)} \oplus A^{(j)}$ of two independent uniform matrices in $\FF_2^{d \times P}$ is itself uniform. To show that a uniform binary matrix is often full rank, consider building such a matrix row-by-row, $v_1,v_2,\ldots$. Because the span of $\ell-1$ binary vectors contains at most $2^{\ell-1}$ binary vectors, we have that
\[\Pr[v_\ell \in \operatorname{span}\{v_1, \dots, v_{\ell-1}\}] \leq \frac{2^{\ell-1}}{2^P} = 2^{\ell - 1 - P}.\]
Taking a union bound over $d$ rows yields that
\[
  \PR\bigl[\operatorname{rank}(A^{(i)} \oplus A^{(j)}) < d\bigr]
  \;\le\; \sum_{\ell=1}^{d} 2^{\ell-1-P} \;<\; n \cdot 2^{-P}
\]
for each pairwise difference matrix. A final union bound over the $\binom{N}{2}$ difference matrices gives the desired result:
\begin{align*}
  \Pr[\exists i \neq j: \operatorname{rank}(A^{(i)} \oplus A^{(j)}) < d] \leq \frac{N^2}{2} \cdot n\cdot 2^{-P} \leq \frac{1}{2n} \leq 1/2.
\end{align*}
\end{proof}

We use \Cref{lemma:measure-gap} to show that a fixed retrieval rule can only beat baseline recall on a small fraction of intents.
In the proof below we give the learner the stronger training observation
$(q,y)$ with $y=Aq$. 
This is stronger than gold-set supervision, $(q,G(q))$, 
since any statistical query $\psi(q,G)$ can be simulated by
$\widetilde{\psi}(q,y)=\psi(q,B(y;r))$, with the same expectation and
tolerance. Thus, a lower bound given $(q,y)$ also
holds under gold-set supervision.

\measureGap*
\begin{proof}
  Write the likelihood ratio $L_A(q,y) = \frac{D_A(q,y)}{\mathcal{N}_0(q,y)} = n \cdot \textbf{1}\{y = Aq\}$ and let $g_A = L_A - 1$. Then we have that $\E_{D_A}[\varphi] = \E_{\mathcal{N}_0}[L_A \varphi]$, so
  \[\Delta_A(\varphi) = \E[(L_A - 1) \varphi] = \E[g_A \varphi].\]
  Moreover, because $\E_{\mathcal{N}_0}[g_A] = 0$, we can additionally write $\E_{\mathcal{N}_0}[g_A \varphi] = \E_{\mathcal{N}_0}[g_A (\varphi - \E_{\mathcal{N}_0}[\varphi])]$. The remainder of the proof is straightforward: we first show that $\{g_A/\sqrt{n-1}\}_{A \in \mathcal{A}}$ is a set of orthonormal random variables, then apply Bessel's inequality to directly obtain the desired result.

  Toward that end, note that for any intents $A,A' \in \mathcal{A}$,
\[
  \EE_{\mathcal{N}_0}[L_A L_{A'}]
  = n^2 \cdot \PR_{\mathcal{N}_0}\bigl[ y = Aq,  y = A'q \bigr]
  = n \cdot \PR_{q \sim \mathcal{Q}}\bigl[ Aq = A'q \bigr].
\]
The first inequality follows by definition and the second from the fact that $y$ is uniform and independent of $q$ under $\mathcal{N}_0$. By condition \eqref{eq:chance-collision}, the right-hand-side is equal to $n$ when $A=A'$ and $1$ otherwise. Thus,
\[
  \E_{\mathcal{N}_0}[g_A g_{A'}]
  = \EE_{\mathcal{N}_0}[L_A L_{A'}] - 1
  = (n-1)\,\mathbf{1}\{A = A'\} ,
\]
and $\{ g_A / \sqrt{n-1} \}_{A \in \mathcal{A}}$ is orthonormal. Finally, applying Bessel's inequality yields
\begin{equation*}
  \sum_{A \in \mathcal{A}} \Delta_A(\varphi)^2
  \;=\sum_{A \in \mathcal{A}}\E_{\mathcal{N}_0}[g_A(\varphi - \E_{\mathcal{N}_0}[\varphi])]^2\; \leq \;(n-1)\,\Var_{\mathcal{N}_0}(\varphi).
\end{equation*}
Each intent $A$ with $|\Delta_A(\varphi)| > \tau$ contributes at least $\tau^2$ to the sum on the left-hand side, and therefore there can be at most $\frac{(n-1) \cdot \operatorname{Var}_{\mathcal{N}_0}(\varphi)}{\tau^2}$ such intents.
\end{proof}

\intentCov*
\begin{proof}
Let $G(y)$ denote the gold set around center $y$ and consider the query $\varphi_R(q,y) \;=\; \rec\bigl(R(q),\, G(y)\bigr) \in [0,1]$. Under $D_A$, the center is $y = Aq$ and $\E_{D_A}[\varphi_R]
= \E_{q \sim \mathcal{Q}}\bigl[ \rec(R(q), G_A(q)) \bigr].$ Under $\mathcal{N}_0$, the gold center $y$ is
uniform and independent of $q$, so each of the $k$ documents in $R(q)$ lie in $G(y)$ with probability $k/n$. Thus,
\[\E_{\mathcal{N}_0}[\varphi_R] = \frac{1}{|G(y)|} \cdot \sum_{d_i \in R(q)} \Pr[d_i \in G(y)] = \frac{1}{k} \cdot \frac{k^2}{n} = k/n.\]
We've shown that $\Delta_A(\varphi_R) = \E_{D_A}[\rec(R(q), G_A(q))] - \frac{k}{n}$. Moreover, $\varphi_R^2 \le \varphi_R$ since $\varphi_R \in [0,1]$, and so
\[
  \Var_{\mathcal{N}_0}(\varphi_R)
  \;\le\; \E_{\mathcal{N}_0}[\varphi_R] - \Bigl(\frac{k}{n}\Bigr)^{\!2}
  \;=\; \frac{k}{n}\Bigl(1 - \frac{k}{n}\Bigr).
\]
An application of \Cref{lemma:measure-gap} $\tau = \varepsilon$ yields the result.
\end{proof}

\section{Additional Experimental Details}
\subsection{Setup} \label{app:exp-setup}
\noindent\textbf{Benchmarks.} All datasets except \textsf{fiqa}, \textsf{nfcorpus} and \textsf{scifact} are evaluation-only and subject to a 70/30 split of the query set for training.\\

\begin{tabular}{llrrrrr}
\toprule
 &  &  & \multicolumn{3}{c}{Queries} &  \\
\cmidrule(lr){4-6}
& Category & Documents & total & train & eval & mean $|G(q)|$ \\
\midrule
\multirow{1}{*}{\textbf{LIMIT}} & limit & 50,000 & 1000 & 700 & 300 & 2.0 \\
\midrule
\multirow{12}{*}{\textbf{BRIGHT}} & biology & 57,359 & 103 & 72 & 31 & 3.6 \\
 & earth science & 121,249 & 116 & 81 & 35 & 5.0 \\
 & economics & 50,220 & 103 & 72 & 31 & 7.8 \\
 & psychology & 52,835 & 101 & 70 & 31 & 6.9 \\
 & robotics & 61,961 & 101 & 70 & 31 & 5.1 \\
 & stackoverflow & 107,081 & 117 & 81 & 36 & 4.1 \\
 & sustainable living & 60,792 & 108 & 75 & 33 & 5.3 \\
 & leetcode & 413,932 & 142 & 99 & 43 & 1.8 \\
 & pony & 7,894 & 112 & 78 & 34 & 19.8 \\
 & aops & 188,002 & 111 & 77 & 34 & 4.7 \\
 & theoremqa theorems & 23,839 & 76 & 53 & 23 & 2.0 \\
 & theoremqa questions & 188,002 & 194 & 135 & 59 & 2.3 \\
\midrule
\multirow{8}{*}{\textbf{BEIR}} 
 & fiqa & 57,638 & 6148 & 5500 & 648 & 2.6 \\
 & nfcorpus & 3,633 & 2913 & 2590 & 323 & 42.2 \\
 & quora & 522,931 & 10000 & 7000 & 3000 & 1.6 \\
 & scidocs & 25,657 & 1000 & 700 & 300 & 4.9 \\
 & scifact & 5,183 & 1109 & 809 & 300 & 1.1 \\
 & trec-covid & 171,332 & 50 & 35 & 15 & 493.5 \\
 & webis-touche2020 & 382,545 & 49 & 34 & 15 & 19.0 \\
\bottomrule
\end{tabular}

\\\\
\noindent\textbf{Hyperparameters.}
Embedding model fine tuning is done using LoRA($r=16$, $\alpha=32$, dropout $0.05$) with the AdamW optimizer under a cosine schedule with $100$ warmup steps, weight decay $0.01$, and default learning rate $10^{-4}$, for $3000$ steps. We use a multi-positive InfoNCE loss with temperature $0.05$. \Cref{tab:hparams} gives full encoder configurations, including LoRA targets. BM25 is a pure Python Okapi implementation with $k_1=0.9, b=0.4$.\\

\begin{table}[tb]
\centering
\small
\setlength{\tabcolsep}{6pt}
\renewcommand{\arraystretch}{1.2}
\begin{tabular}{@{}l l l l@{}}
\toprule
& \textbf{DistilBERT} & \textbf{Qwen3-Embedding-4B} & \textbf{GTE-ModernColBERT} \\
\midrule
Checkpoint
  & \texttt{sentence-transformers/}
  & \texttt{Qwen/}
  & \texttt{lightonai/} \\
  & \texttt{msmarco-distilbert-}
  & \texttt{Qwen3-Embedding-4B}
  & \texttt{GTE-ModernColBERT-} \\
  & \texttt{base-v3}
  &
  & \texttt{v1} \\
Representation
  & single-vector & single-vector & multi-vector (MaxSim) \\
Pooling
  & mean          & last token    & none (per token) \\
Truncation
  & 350 tok.      & 512 tok.      & 300 tok.\ (queries 32)  \\
LoRA targets
  & \texttt{q\_lin}, \texttt{k\_lin},
  & \texttt{q\_proj}, \texttt{k\_proj},
  & \texttt{Wqkv}, \texttt{Wo} \\
  & \texttt{v\_lin}, \texttt{out\_lin}
  & \texttt{v\_proj}, \texttt{o\_proj}
  & \\
\midrule
Training batch size
  & 32  & 32  & 8 \\
Encoding batch size
  & 256 & 256 & 64 \\
\bottomrule
\end{tabular}
\caption{Encoder configurations. The default Qwen3 query prompt is
\texttt{"Instruct: Given a search query, retrieve relevant passages that
answer the query\textbackslash nQuery: "}.}
\label{tab:hparams}
\end{table}

\noindent\textbf{Compute environment.}
All experiments (document indexing, query encoder fine-tuning, geometric ceiling lower bound computation) were run on an NVIDIA HGX H200 node with 8 GPUs (141GB memory each), 2× Intel Xeon Platinum 8470 CPUs (104 cores total), and 1T of RAM.
\subsection{Full results} \label{app:results}
\Cref{tab:capacity} reports the ranking-SVM geometric ceiling lower bound for frozen document indexes embedded via \textsf{DistilBERT} and \textsf{Qwen3-Embedding-4B}, respectively. \Cref{tab:pretrained} reports out-of-the-box recall for each of the all four retrieval methods we consider. \Cref{tab:finetuned} reports recall of fine-tuned single-vector and multi-vector retrieval models, with performance measured on both training and evaluation queries.
\begin{table}[h]\centering\small{\footnotesize\setlength{\tabcolsep}{3pt}
\begin{tabular}{ll rr c rr c}
\toprule
 &  & \multicolumn{2}{c}{\makebox[0pt][c]{DistilBERT}} & \rule{0pt}{0pt}\hspace{20pt} & \multicolumn{2}{c}{\makebox[0pt][c]{Qwen3-Embedding-4B}} & \rule{0pt}{0pt}\hspace{12pt} \\
\cmidrule(lr){3-4} \cmidrule(lr){6-7}
& Category & train & eval & & train & eval & \\
\midrule
\multirow{1}{*}{\textbf{LIMIT}} & limit & 1.000 & 1.000 & & 1.000 & 1.000 &  \\
\midrule
\multirow{12}{*}{\textbf{BRIGHT}} & biology & 0.979 & 0.962 & & 0.980 & 0.954 &  \\
 & earth science & 0.980 & 0.997 & & 0.985 & 0.999 &  \\
 & economics & 0.983 & 1.000 & & 0.984 & 0.999 &  \\
 & psychology & 0.999 & 0.996 & & 0.995 & 0.990 &  \\
 & robotics & 0.990 & 0.968 & & 0.990 & 0.968 &  \\
 & stackoverflow & 0.847 & 0.833 & & 0.902 & 0.910 &  \\
 & sustainable living & 0.990 & 0.957 & & 0.998 & 0.956 &  \\
 & leetcode & 0.987 & 1.000 & & 0.993 & 1.000 &  \\
 & pony & 1.000 & 1.000 & & 1.000 & 1.000 &  \\
 & aops & 1.000 & 1.000 & & 1.000 & 1.000 &  \\
 & theoremqa theorems & 1.000 & 1.000 & & 1.000 & 1.000 &  \\
 & theoremqa questions & 0.994 & 0.996 & & 0.998 & 1.000 &  \\
\midrule
\multirow{8}{*}{\textbf{BEIR}} 
 & fiqa & 0.995 & 0.996 & & 0.994 & 0.995 &  \\
 & nfcorpus & 0.975 & 0.975 & & 0.982 & 0.986 &  \\
 & quora & 0.989 & 0.989 & & 0.984 & 0.984 &  \\
 & scidocs & 0.983 & 0.980 & & 0.980 & 0.976 &  \\
 & scifact & 1.000 & 1.000 & & 1.000 & 1.000 &  \\
 & trec-covid & 0.539 & 0.493 & & 0.767 & 0.737 &  \\
 & webis-touche2020 & 0.858 & 0.860 & & 0.907 & 0.938 &  \\
\bottomrule
\end{tabular}}

\caption{Geometric capacity by category, as measured by the ranking SVM.}
\label{tab:capacity}\end{table}

\begin{table}[h]\centering\small{\footnotesize\setlength{\tabcolsep}{3pt}
\begin{tabular}{ll cccc}
\toprule
& Category & BM25 & DistilBERT & Qwen3-Embedding-4B & GTE-ModernColBERT \\
\midrule
\multirow{1}{*}{\textbf{LIMIT}} & limit & 0.917 & 0.002 & 0.015 & 0.525 \\
\midrule
\multirow{12}{*}{\textbf{BRIGHT}} & biology & 0.156 & 0.016 & 0.073 & 0.078 \\
 & earth science & 0.201 & 0.078 & 0.233 & 0.168 \\
 & economics & 0.186 & 0.074 & 0.179 & 0.143 \\
 & psychology & 0.087 & 0.073 & 0.210 & 0.102 \\
 & robotics & 0.046 & 0.045 & 0.038 & 0.104 \\
 & stackoverflow & 0.095 & 0.004 & 0.034 & 0.036 \\
 & sustainable living & 0.104 & 0.048 & 0.061 & 0.051 \\
 & leetcode & 0.048 & 0.048 & 0.062 & 0.095 \\
 & pony & 0.109 & 0.080 & 0.004 & 0.108 \\
 & aops & 0.038 & 0.014 & 0.034 & 0.008 \\
 & theoremqa theorems & 0.000 & 0.000 & 0.189 & 0.000 \\
 & theoremqa questions & 0.034 & 0.008 & 0.056 & 0.008 \\
\midrule
\multirow{8}{*}{\textbf{BEIR}} 
 & fiqa & 0.149 & 0.182 & 0.455 & 0.147 \\
 & nfcorpus & 0.172 & 0.143 & 0.242 & 0.052 \\
 & quora & 0.607 & 0.753 & 0.772 & 0.210 \\
 & scidocs & 0.106 & 0.100 & 0.194 & 0.089 \\
 & scifact & 0.541 & 0.405 & 0.632 & 0.552 \\
 & trec-covid & 0.236 & 0.153 & 0.480 & 0.094 \\
 & webis-touche2020 & 0.349 & 0.135 & 0.233 & 0.045 \\
\bottomrule
\end{tabular}}

\caption{Expected recall@$|G(q)|$ of pre-trained models, without fine-tuning.}
\label{tab:pretrained}\end{table}

\begin{table}[h]\centering\small{\footnotesize\setlength{\tabcolsep}{3pt}
\begin{tabular}{ll rr c rr c rr c}
\toprule
 &  & \multicolumn{2}{c}{\makebox[0pt][c]{DistilBERT}} & \rule{0pt}{0pt}\hspace{39pt} & \multicolumn{2}{c}{\makebox[0pt][c]{Qwen3-Embedding-4B}} & \rule{0pt}{0pt}\hspace{39pt} & \multicolumn{2}{c}{\makebox[0pt][c]{GTE-ModernColBERT}} & \rule{0pt}{0pt}\hspace{12pt} \\
\cmidrule(lr){3-4} \cmidrule(lr){6-7} \cmidrule(lr){9-10}
& Category & train & eval & & train & eval & & train & eval & \\
\midrule
\multirow{1}{*}{\textbf{LIMIT}} & limit & 0.439 & 0.045 & & 1.000 & 0.068 & & 1.000 & 0.998 &  \\
\midrule
\multirow{12}{*}{\textbf{BRIGHT}} & biology & 0.972 & 0.105 & & 0.972 & 0.189 & & 0.975 & 0.240 &  \\
 & earth science & 0.973 & 0.178 & & 0.982 & 0.274 & & 0.910 & 0.243 &  \\
 & economics & 0.984 & 0.103 & & 0.995 & 0.276 & & 0.971 & 0.197 &  \\
 & psychology & 0.992 & 0.021 & & 0.996 & 0.151 & & 0.970 & 0.119 &  \\
 & robotics & 0.972 & 0.037 & & 0.995 & 0.131 & & 0.964 & 0.113 &  \\
 & stackoverflow & 0.836 & 0.031 & & 0.900 & 0.103 & & 0.712 & 0.044 &  \\
 & sustainable living & 0.987 & 0.069 & & 0.995 & 0.088 & & 0.939 & 0.102 &  \\
 & leetcode & 0.991 & 0.070 & & 1.000 & 0.116 & & 0.745 & 0.078 &  \\
 & pony & 0.930 & 0.740 & & 0.955 & 0.790 & & 0.998 & 0.801 &  \\
 & aops & 0.878 & 0.148 & & 0.994 & 0.310 & & 0.758 & 0.062 &  \\
 & theoremqa theorems & 1.000 & 0.159 & & 1.000 & 0.304 & & 1.000 & 0.040 &  \\
 & theoremqa questions & 0.968 & 0.112 & & 0.998 & 0.213 & & 0.916 & 0.075 &  \\
\midrule
\multirow{8}{*}{\textbf{BEIR}} 
 & fiqa & 0.222 & 0.180 & & 0.824 & 0.363 & & 0.303 & 0.285 &  \\
 & nfcorpus & 0.211 & 0.154 & & 0.686 & 0.226 & & 0.276 & 0.218 &  \\
 & quora & 0.805 & 0.754 & & 0.942 & 0.770 & & 0.704 & 0.694 &  \\
 & scidocs & 0.325 & 0.080 & & 0.934 & 0.192 & & 0.198 & 0.098 &  \\
 & scifact & 0.990 & 0.661 & & 1.000 & 0.766 & & 0.975 & 0.760 &  \\
 & trec-covid & 0.489 & 0.231 & & 0.583 & 0.470 & & 0.365 & 0.284 &  \\
 & webis-touche2020 & 0.637 & 0.142 & & 0.747 & 0.203 & & 0.167 & 0.142 &  \\
\bottomrule
\end{tabular}}

\caption{Expected recall@$|G(q)|$ after query-side fine-tuning.}
\label{tab:finetuned}\end{table}

\end{document}